\documentclass[english,draftcls,12pt,onecolumn]{IEEEtran}
\usepackage[T1]{fontenc}
\usepackage[latin9]{inputenc}
\usepackage{verbatim}
\usepackage{amsmath}
\usepackage{setspace}
\usepackage{amssymb}
\usepackage{esint}
\makeatletter
\newtheorem{prop}{Proposition}
\newtheorem{thm}{Theorem}
\newtheorem{remrk}{Remark}
\newtheorem{lemma}{Lemma}

\usepackage{cite}
\usepackage{array}
\usepackage{mdwmath}
\usepackage{mdwtab}

\newcommand{\mythanks}{\thanks{The author is with the Control \& Dynamical Systems department, Division of Engineering and Applied Sciences, California Institute of Technology. Address: MC 107-81, 1200 E. California Blvd., 91125, Pasadena, CA.
        E-mail: andrea@cds.caltech.edu}}
\usepackage{xcolor}

\usepackage[mathscr]{eucal}

\newcommand{\rdist}{d}
\newcommand{\manifold}{\mathscr{M}}
\newcommand{\metric}{{ g}}

\newcommand{\Sym}{\mathscr{S}}
\newcommand{\SymP}{\mathscr{P}}

\newcommand{\mlog}{\mathrm{Log}}
\newcommand{\mexp}{\mathrm{Exp}}

\newcommand{\opt}{\mathrm{opt}}

\newcommand{\pess}{\mathrm{pes}}

\newcommand{\astring}{\omega}
\newcommand{\stringset}{\{0,1\}^{\mathbb{N}}}

\newcommand{\ifs}{\mathcal{S}}
\newcommand{\ifsopt}{\ifs_{\opt}}
\newcommand{\ifspess}{\ifs_{\pess}}

\newcommand{\gaussdist}{\mathcal{G}}
\newcommand{\Gauss}{\mathscr{G}}
\newcommand{\GaussO}{\mathscr{G}_0}

\newcommand{\fg}{g}
\newcommand{\fh}{h}
\newcommand{\fgopt}{{g}_{\opt}}
\newcommand{\fhopt}{{h}_{\opt}}
\newcommand{\fgpess}{{g}_{\pess}}
\newcommand{\fhpess}{{h}_{\pess}}

\newcommand{\vecsym}[1]{\boldsymbol{#1}}

\newcommand{\ve}{\vecsym{e}}
\newcommand{\vx}{\vecsym{x}}
\newcommand{\vy}{\vecsym{y}}
\newcommand{\vz}{\vecsym{z}}
\newcommand{\vtheta}{\vecsym{\theta}}
\newcommand{\vzmean}{\vecsym{\mu}}
\newcommand{\vzcov}{\vecsym{\Sigma}}

\newcommand{\veta}{\vecsym{\eta}}
\newcommand{\vzero}{\vecsym{0}}
\newcommand{\se}{e}

\newcommand{\ex}{\mathbb{E}}
\newcommand{\exi}{\mathbb{M}}
\newcommand{\pr}{\mathbb{P}}

\newcommand{\m}[1]{\boldsymbol{\mathrm{#1}}}
\newcommand{\mA}{\m{A}}
\newcommand{\mB}{\m{B}}
\newcommand{\mC}{\m{C}}

\newcommand{\mI}{\m{I}}
\newcommand{\mM}{\m{M}}
\newcommand{\mQ}{\m{Q}}
\newcommand{\mX}{\m{X}}
\newcommand{\mY}{\m{Y}}

\newcommand{\vnoiseinput}{\boldsymbol{\omega}}
\newcommand{\vnoise}{\boldsymbol{\epsilon}}

\newcommand{\fim}{\boldsymbol{\mathcal{I}}}

\newcommand{\mP}{\boldsymbol{\mathrm{P}}}
\newcommand{\mPinf}{\boldsymbol{\mathrm{P}}_{\!\infty}}

\newcommand{\mPpess}{\boldsymbol{\mathrm{P}}_{\!\pess}}

\newcommand{\mWpess}{\m{W}_{\!\pess}}
\newcommand{\sWopt}{\mathrm{W}_{\!\opt}}
\newcommand{\sWpess}{\mathrm{W}_{\!\pess}}

\newcommand{\sP}{\mathrm{P}}
\newcommand{\sPinf}{\mathrm{P}_{\!\infty}}
\newcommand{\sPopt}{\mathrm{P}_{\!\opt}}
\newcommand{\sPpess}{\mathrm{P}_{\!\pess}}
\newcommand{\sQ}{\mathrm{Q}}
\newcommand{\sfim}{\mathcal{I}}

\newcommand{\arr}{\gamma}
\newcommand{\probArr}{\overline{\gamma}}
\newcommand{\probArrCrit}{\probArr_c}
\newcommand{\probArrCritO}{{\probArr'}_c}

\makeatother

\usepackage{babel}

\begin{document}

\title{Geometric remarks on\\
Kalman filtering with intermittent observations }

\author{Andrea Censi \mythanks}
\maketitle
\begin{abstract}
Sinopoli \emph{et al.} (TAC, 2004) considered the problem of optimal
estimation for linear systems with Gaussian noise and intermittent
observations, available according to a Bernoulli arrival process.
They showed that there is a {}``critical'' arrival probability of
the observations, such that under that threshold the expected value
of the covariance matrix (i.e., the quadratic error) of the estimate
is unbounded. Sinopoli \emph{et al.}, and successive authors, interpreted
this result implying that the behavior of the system is qualitatively
different above and below the threshold. This paper shows that this
is not necessarily the only interpretation. In fact, the critical
probability is different if one considers the average error instead
of the average quadratic error. More generally, finding a meaningful
{}``average'' covariance is not as simple as taking the algebraic
expected value. A rigorous way to frame the problem is in a differential
geometric framework, by recognizing that the set of covariance matrices
(or better, the manifold of Gaussian distributions) is not a flat
space, and then studying the intrinsic Riemannian mean. Several metrics
on this manifold are considered that lead to different critical probabilities,
or no critical probability at all.
\end{abstract}
%\begin{IEEEkeywords} Kalman filtering, intrinsic estimation \end{IEEEkeywords}

\section{Introduction}

%\IEEEPARstart{T}{he}  

The Kalman filter was conceived in the 1960s~\cite{kalman1960} and
found immediate use at the forefront of engineering~\cite{McGee:1985ne}%
\begin{comment}
, computational application of the theory of Wiener filtering
\end{comment}
{}. For the successive decades, the state-space approach of the Kalman
filter was the tool of choice for many filtering and tracking problems,
both in its algebraically equivalent formulations (e.g., Information
filter~\cite{maybeck79}, square root and {}``array'' algorithms~\cite{kailath00linear})
and its extensions to nonlinear problems (e.g., extended Kalman filter
(EKF), unscented Kalman filter), only recently giving way to Monte
Carlo methods (particle filters). %
\begin{comment}
From:\%Carlson, Neal A. (1973) Fast triangular formulation of the
square root filter. \% AIAA Journal, 11, pp 1259-1265. 

Square root: Potter/LEM for the apollo
\end{comment}
{}

In recent years, in many engineering fields, estimation problems have
been considered where the availability or observations, or their structure,
is subject to random phenomena, and one is interested in characterizing
the {}``average accuracy''. For example, in robotics, the EKF is
used in problems such as Simultaneous Localization and Mapping (SLAM);
the observations structure depends on the landmark configuration,
which is unknown a priori, yet it is of interest to study average
accuracy results~\cite{mourikis08characterization}.

In the control literature, random observations can model packet drops,
which is one of the important phenomena in network-based estimation
and control. Sinopoli \emph{et al.~}\cite{sinopoli04kalman}\emph{
}considered the problem of Kalman filtering when the observations
are available intermittently with Bernoulli probability. They showed
that there exists a critical value of the arrival probability such
that, under that threshold, the expected value of the error covariance
matrix is unbounded. Other successive papers improved on the same
results by better characterizing the critical probability or considering
non-independent packet drops~\cite{Huang2007Stability,mo08characterization,plarre09kalman,kar09kalman}%
\begin{comment}
\cite{sinopoli09}
\end{comment}
{}. 

The way the result of Sinopoli \emph{et al. }is often interpreted
is that the system has a qualitatively different behavior\emph{ }above
and below the critical probability. The purpose of this note is to
show that this is not necessarily the only interpretation. A motivating
example is given in Section~\ref{sec:Motivating-example}. Considering
the expected value of the covariance is equivalent to considering
the expected value of the squared error norm~$\ex\{\left\Vert \ve\right\Vert _{2}^{2}\}$.
If one instead considers the error norm $\ex\{\left\Vert \ve\right\Vert _{2}\}$,
which is equivalent to considering the expected value of the standard
deviation, a different --- and lower --- critical probability is obtained.
This raises doubts about the significance of Sinopoli \emph{et al.}'s
critical probability. More generally, what is critical is the way
one defines the {}``average'' uncertainty. Because the operation
of expected value is not invariant to change of coordinates, the result
is different if one averages the covariances, the standard deviations,
or the information matrices: in general, $\ex\{\mP\}\neq\ex\{\sqrt{\mP}\}^{2}\neq\ex\{\mP^{-1}\}^{-1}$.
This paper advocates a geometric point of view. The basic assumption
is that covariance matrices are only a particular choice of coordinates
to represent Gaussian distributions, which is a Riemannian manifold
with a very rich structure. Section~\ref{sec:Means-on-manifolds}
deals with how to extend the idea of {}``mean'' to Riemannian manifolds,
and how that depends on the choice of a metric. Section~\ref{sec:Different-metrics-for}
discusses several metrics one can use for the manifold of Gaussian
distributions. After the obvious metrics are discussed (which lead
to averaging covariances, information matrices, etc.), a non-trivial
Riemannian metric is introduced that is shown to be the most most
natural when dealing with Gaussian distributions, or, in general,
when considering the intrinsic properties of the set of positive definite
matrices. These different metrics lead to different critical probabilities,
or no critical probability at all.

%\medskip

\paragraph*{\textmd{Notation}}

All matrices are assumed to be real. Let $\mA^{*}$ be the transpose
of the matrix~$\mA$, let $\mbox{Tr}(\mA)$ be its trace, and $\{\lambda_{i}(\mA)\}$
its eigenvalues. Let $\mathrm{GL}(n)$ be the set of $n\times n$
invertible matrices; let $\mathrm{O}(n)$ be the set of orthogonal
matrices; let $\Sym(n)$ be the set of symmetric $n\times n$ matrices;
and let $\SymP(n)\subset\Sym(n)$ be the set of positive definite
matrices. Let $\Gauss(n)$ be the manifold of Gaussian distributions
on $\mathbb{R}^{n}$, and $\GaussO(n)\subset\Gauss(n)$ the submanifold
of Gaussian distributions with mean~$\vzero$. An element of $\Gauss(n)$
is denoted as $\gaussdist(\boldsymbol{\mu},\mP)$, where the mean
$\boldsymbol{\mu}\in\mathbb{R}^{n}$ and the covariance $\mP\in\SymP(n)$
serve as coordinates on $\Gauss(n)$. Let~$\left\Vert \cdot\right\Vert $
be the operator norm ($\left\Vert \mA\right\Vert ^{2}=\lambda_{\max}(\mA\mA^{*})$),
and let~$\left\Vert \cdot\right\Vert _{F}$ be the Frobenius norm
($\left\Vert \mA\right\Vert _{F}^{2}=\mbox{Trace}(\mA\mA^{*})$).
For $\mP\in\SymP(n)$, let~$\sqrt{\mP}$ be the unique matrix in~$\SymP(n)$
such that $(\sqrt{\mP})^{2}=\mP$. %
\begin{comment}
Let $\mexp(\mA)$ be the matrix exponential of $\mA$: $\mexp(\mA)\triangleq\sum_{k=0}^{\infty}(\mA^{k}/k!)$
and let $\mlog=\mexp^{-1}$.
\end{comment}
{} All inequalities between matrices are to be interpreted in the L\"{o}wner
partial order: $\mP_{1}\geq\mP_{2}$ iff $\mP_{1}-\mP_{2}$ is semidefinite
positive.

\section{Motivating example\label{sec:Motivating-example}}

Consider the discrete-time linear dynamical system\begin{eqnarray*}
\vx(k+1) & = & \mA\,\vx(k)+\mB\,\vnoiseinput(k),\\
\vy(k) & = & \mC\,\vx(k)+\vnoise(k),\end{eqnarray*}
with $\vx\in\mathbb{R}^{n}$, $\vnoiseinput\in\mathbb{R}^{p}$, $\vy\in\mathbb{R}^{q}$,
$ $and~$\mA$,~$\mB$,~$\mC$ real matrices of appropriate sizes.
Assume $\vnoiseinput(k)$ and $\vnoise(k)$ are white Gaussian sequences
with zero mean and covariance matrix equal to the identity, and that
the initial prior for~$\vx(0)$ is Gaussian with mean~$\hat{\vx}(0)$
and covariance $\mP(0)$. Moreover, assume that the observations are
available randomly, i.e., one has available the observations $\vy'(k)=\arr(k)\vy(k)$,
where~$\arr(k)$ is a sequence of independent Bernoulli random variables,
such that $\pr(\{\arr(k)=1\})=\probArr$ and $\pr(\{\arr(k)=0\})=1-\probArr$.
The conditional estimate of $\vx(k)$, given the available observations
until time~$k$ is still Gaussian~\cite{sinopoli04kalman}, and
is indicated by the mean~$\hat{\vx}(k)$ and the covariance~$\mP(k)$.
Define the error estimate $\ve(k)\triangleq\hat{\vx}(k)-\vx(k)$.
Then~$\ve(k)$ has a Gaussian distribution with mean~$\vzero$ and
covariance~$\mP(k)$. This is the setup considered in~\cite{sinopoli04kalman}
and is henceforth called Linear/Gaussian/Bernoulli (LGB); the name
{}``Kalman filtering'' is not used because the results are independent
of the particular representation of the optimal filter.%
\begin{comment}
%
\footnote{{}``Kalman filtering'' would be misleading because the results here
are independent of the particular implementation of the optimal estimator.%
}
\end{comment}
{} 

Let $\mQ\triangleq\mB\mB^{*}$ and $\fim\triangleq\mC^{*}\mC$. If
the observations are always available ($\probArr=1$), the evolution
of $\mP(k)$ is deterministic and obeys the recursion%
\footnote{Note that this paper uses the posterior covariance matrix ($\mP(k)=\mP_{k|k}=\mbox{cov}(\hat{\vx}(k)-\vx(k)|\vy'(1),\dots,\vy'(k)$).
The a-priori covariance $\mP_{k|k-1}$ and a-posteriori $\mP_{k|k}$
are linked by the simple relation $\mP_{k|k-1}=\mA\mP_{k-1|k-1}\mA^{*}+\mQ$
hence there is no loss of generality for investigating the boundedness
of the stationary distribution: $\ex\{\mP_{k|k-1}\}$ is bounded if
and only if $\ex\{\mP_{k|k}\}$ is. Using the posterior covariance
matrix seems a better choice for LGB filtering because, when written
with information matrices ($\mY\mapsto(\mA\mY^{-1}\mA^{*}+\mQ)^{-1}+\fim$)
the difference between the maps $\fg$ and $\fh$ is the constant
term~$\fim$~\cite{censi09fractals}.%
}\begin{equation}
\fg:\mP\mapsto\left((\mA\mP\mA^{*}+\mQ)^{-1}+\fim\right)^{-1}.\label{eq:gsimple}\end{equation}
If $(\mA,\mB)$ is stabilizable and $(\mA,\mC)$ is detectable, then~$g$
has a fixed point $\mPinf$ to which~$\mP(k)$ tends regardless of
the initial value~$\mP(0)$~\cite{maybeck79}. The Kalman filter
and analogous variants implement the recursion with different representations
for~$\mP$, and faster and more numerically stable algorithms than~\eqref{eq:gsimple},
which is used in the present analysis for convenience and compactness
(apply one of the matrix inversion lemmas to obtain the usual Riccati
recursion). 

If $\probArr\in(0,1)$, the evolution of $\mP(k)$ is not deterministic
anymore. A convenient way to represent the evolution of~$\mP(k)$
is in the form of an Iterated Function System~\cite{edgar98integral}\begin{equation}
\ifs=\begin{cases}
g: & \mP\mapsto\left((\mA\mP\mA^{*}+\mQ)^{-1}+\fim\right)^{-1},\qquad p_{g}=\probArr,\\
h: & \mP\mapsto\mA\mP\mA^{*}+\mQ,\qquad\,\,\,\quad\qquad\qquad p_{h}=1-\probArr.\end{cases}\label{eq:ifs}\end{equation}
A stationary distribution for $\mP$ exists for all values of~$\probArr$
(under much more general conditions than Bernoulli observations)~\cite{bougerol93kalman}.%
\begin{comment}
; depending on the parameters of the system $\mA$, $\mQ$ and $\fim$,
but not on~$\probArr$, the cdf of $\mP$ could be a singular function
with a fractal support~\cite{censi09fractals}.
\end{comment}
{} In the following, the stationary distribution is referred to as the
Linear/Gaussian/Bernoulli (LGB) distribution, and~$\mP$ refers to
the stationary variable.

Sinopoli \emph{et al.}~\cite{sinopoli04kalman} showed that there
is a threshold~$\probArrCrit$ such that, for $\probArr<\probArrCrit$,
$\ex\{\mP\}$ is unbounded. This is equivalent to say that the square
of the Euclidean norm of the estimation error, $\ex\{\left\Vert \ve\right\Vert _{2}^{2}\}$,
is unbounded. The threshold~$\probArrCrit$ depends non trivially
on the parameters of the system, and a precise characterization is
object of current research~\cite{sinopoli04kalman,Huang2007Stability,mo08characterization,plarre09kalman}%
\begin{comment}
\cite{sinopoli09}
\end{comment}
{}. 

The way this result is often interpreted is that for~$\probArr\geq\probArrCrit$
the behavior of the system\emph{ }is qualitatively different than
for $\probArr<\probArrCrit$, and therefore~$\probArrCrit$ is called
{}``critical'' probability. However, if one considers another measure
of performance, the critical probability changes, as described by
the following result.

%
\begin{comment}
In Section II, it is shown that, in the scalar case, if one chooses
$\ex\{|\ve|\}$ instead of $\ex\{|\ve|^{2}\}$ as a performance metric,
the critical value is different. Section III takes a step back and
shows how one can define means on a metric manifolds. Finally, Section
IV shows that another metric for covariance matrices shows that there
is, really, no critical value.
\end{comment}
{}
\begin{prop}
\label{pro:1}Consider the problem of LGB filtering in the scalar
case, with $\mA=a>1$, $\fim=\sfim>0$, $\mQ=\sQ\geq0$. Then the
expected squared error $\ex\{\se^{2}\}$ and the expected error $\ex\{|\se|\}$
have two different critical probabilities:
\begin{enumerate}
\item $\ex\{\se^{2}\}$ is bounded if and only if $\probArr>\probArrCrit=1\!-1/a^{2}$.
\item $\ex\{|\se|\}$ is bounded if and only if $\probArr>\probArrCritO=1-1/|a|$.
\end{enumerate}
\end{prop}
\begin{proof}
It is convenient to define two others systems, which are respectively
an {}``optimist'' and a {}``pessimist'' approximations to the
iteration defined by~\eqref{eq:ifs}. %
\begin{comment}
The stationary distribution of~$\sP$ exists for all $\probArr\in(0,1)$~\cite{censi09fractals}. 
\end{comment}
{}The stationary distribution has support in the set $\{\sP|\sP\geq\sPinf\}$.
The optimist approximation simplifies the map~$\fg$ to a constant
by considering the best scenario $\sP=\sPinf$:\begin{equation}
\ifsopt=\begin{cases}
\fgopt: & \sPopt\mapsto\sWopt,\qquad\qquad\quad\sWopt\triangleq\left((a^{2}\sPinf+\sQ)^{-1}+\sfim\right)^{-1},\\
\fhopt: & \sPopt\mapsto a^{2}\sPopt+\sQ.\end{cases}\label{eq:ifsopt}\end{equation}
The pessimist approximation considers the worst case ($\sP\rightarrow\infty$):\begin{equation}
\ifspess=\begin{cases}
\fgpess: & \sPpess\mapsto\sWpess,\qquad\qquad\quad\sWpess\triangleq\sfim^{-1},\\
\fhpess: & \sPpess\mapsto a^{2}\sPpess+\sQ.\end{cases}\label{eq:ifspess}\end{equation}
The two systems are identical except for the {}``reset'' values
$\sWopt$ and $\sWpess$ after an observation is received. It is straightforward
to check that, if $\sPopt(0)=\sP(0)=\sPpess(0)$ and the three systems
see the same sequences of observations, $\sPopt(k)\leq\sP(k)\leq\sPpess(k)$.
Therefore, for the stationary variables, %
\begin{comment}
under the hypothesis that $\sPpess$ and $\sPopt$ have a stationary
distribution (the proof of which is omitted),
\end{comment}
{}$\ex\{\sPopt\}\leq\ex\{\sP\}\leq\ex\{\sPpess\}$ and $\ex\{\sqrt{\sPopt}\}\leq\ex\{\sqrt{\sP}\}\leq\ex\{\sqrt{\sPpess}\}$. 

Obviously $\ex_{\vnoiseinput,\vnoise}\{\se^{2}(k)\}=\sP(k)$, by the
definition of covariance and the fact that $\ex_{\vnoiseinput,\vnoise}\{\se(k)\}=0$.
Moreover, in the case of a Gaussian distribution, one can show that
$\ex_{\vnoiseinput,\vnoise}\{|\se(k)|\}=\sqrt{2/\pi}\sqrt{\sP(k)}$.
Therefore, upper and lower bounds for $\ex\{\se^{2}\}$ and $\ex\{|\se|\}$
can be found as\begin{equation}
\ex\{\sPopt\}\leq\ex\{\se^{2}\}\leq\ex\{\sPpess\},\label{eq:bound_cov}\end{equation}
\begin{equation}
\sqrt{2/\pi}\ex\{\sqrt{\sPopt}\}\leq\ex\{|\se|\}\leq\sqrt{2/\pi}\ex\{\sqrt{\sPpess}\}.\label{eq:bound_sqrt}\end{equation}

The rest of the proof estimates the terms in these expressions and
is inspired by some ideas in~\cite{epstein08probabilistic}. The
pdf for the stationary distribution for the two IFSs~\eqref{eq:ifsopt}--\eqref{eq:ifspess}
can be computed in closed form. Consider, for example, the IFS in~\eqref{eq:ifsopt}.
The value of~$\sPopt$ at time~$k$ can be written in closed form
as a function of~$\tau(k)$, the number of steps that passed without
receiving an observation ($\tau(k)=0$ if the last observation was
received):%
\begin{comment}
\[
\sPopt(k)=\begin{cases}
\sWopt & \tau(k)=0\\
(a^{2})^{\tau(k)}\sWopt+\sum_{i=0}^{\tau(k)-1}(a^{2})^{i}\sQ & \tau(k)\geq1\end{cases}\]

\end{comment}
{}\[
\sPopt(k)=(a^{2})^{\tau(k)}\sWopt+\sum_{i=0}^{\tau(k)-1}(a^{2})^{i}\sQ.\]
Assuming independent arrivals, $\tau(k)$ has the probability distribution
$\pr\left(\left\{ \tau(k)=j\right\} \right)=(1-\probArr)^{j}\probArr$.
The expected value $\ex\{\sPopt\}$ can be computed as $\sum_{j=0}^{\infty}\pr\left(\left\{ \tau(k)=j\right\} \right)\sPopt(\tau(k))$,
giving\[
\ex\{\sPopt\}=\probArr\left(\sWopt+\frac{\sQ}{a^{2}-1}\right)\sum_{j=0}^{\infty}\left[a^{2}(1-\probArr)\right]^{j}-\frac{\sQ}{a^{2}-1}.\]
The series converges, and $\ex\{\sPopt\}$ is bounded, if and only
if $a^{2}(1-\probArr)<1$ (as already proved in~\cite{sinopoli04kalman}).
Analogously, the expected value $\ex\{\sqrt{\sPopt}\}$ can be computed
as\[
\ex\{\sqrt{\sPopt}\}=\probArr\sum_{j=0}^{\infty}\left((1-\probArr)|a|\right)^{j}\sqrt{\left(\sWopt+\frac{\sQ}{a^{2}-1}\right)-\frac{\sQ}{(a^{2})^{j}(a^{2}-1)}}.\]
 The series converges, and $\ex\{\sqrt{\sPopt}\}$ is bounded, if
and only if $|a|(1-\probArr)<1$. 

Because the proof did not rely on the value of $\sWopt$, the same
convergence critical values are valid for the pessimist approximations
$\ex\{\sPpess\}$ and $\ex\{\sqrt{\sPpess}\}$ as well. By taking
into account~\eqref{eq:bound_cov} and \eqref{eq:bound_sqrt}, we
see that $a^{2}(1-\probArr)<1$ is a necessary and sufficient condition
for boundedness of $\ex\{\se^{2}\}$, and likewise $|a|(1-\probArr)<1$
for boundedness of $\ex\{|\se|\}$.
\end{proof}
Because $\probArrCrit>\probArrCritO$, there is a range of values
$(\probArrCritO,\probArrCrit]$ such that $\ex\{|\se|\}$ is bounded,
but $\ex\{\se^{2}\}$ is not. The value $\probArrCritO$ is as least
as {}``critical'' as $\probArrCrit$. The goal of this paper is
not to advocate the use of the boundedness of $\ex\{\left\Vert \ve\right\Vert \}$
rather than $\ex\{\left\Vert \ve\right\Vert ^{2}\}$ as a criterion
of stability; rather, it is of more interest to discuss what are the
assumptions behind using one or the other. Is one {}``intrinsically''
more correct? Other that in Kalman filtering with intermittent observations,
similar questions arise in other problems where one must compute an
{}``average'' accuracy~\cite{mourikis08characterization}. In general,
the expected value is not invariant to change of coordinates, so a
different average accuracy is obtained if one considers the average
of covariances, of standard deviations, or of information matrices. 

Some answers to these questions can be found by setting the problem
in a geometric framework. In particular, instead of considering the
set $\SymP(n)$ of covariance matrices as a subset of $ $$\mathbb{R}^{n\times n}$,
one can consider, more abstractly, the manifold~$\Gauss(n)$ of Gaussian
distributions. The next section shows how the operation of expected
value $\ex\{\cdot\}$ can be generalized to Riemannian manifolds,
such that one can define a {}``Riemannian mean'' $\exi\{\cdot\}$
independently of the choice of coordinates. %
\begin{comment}
consider the Pareto distribution, defined as $Pr(X>x)=\alpha/x^{k}$,
with $\alpha,k>0$. and 
\end{comment}
{}%
\begin{comment}
To appreciate this approach, one should try to forget that covariances
are central in Kalman filters. Consider the system ??. If the initial
estimate of $x$ is Gaussian, then at all successive times the conditional
distribution of the state given the measurement is Gaussian. This
can be verified by direct computation from the Bayesian update:\[
p(x)=\int p(x|x,y)...\]

\end{comment}
{}

\section{Means on Riemannian manifolds\label{sec:Means-on-manifolds}}

Classical mathematical statistics~\cite{CasellaBerger} developed
in the first decades of last century in the context of Euclidean spaces.
Subsequently, it became clear that many applications would benefit
from rigorous coordinate-free approaches to statistics on manifolds.
Examples of such applications and corresponding manifolds include
robotics~\cite{wan06secondorder} (motion groups), shape analysis~\cite{kendall93}
(size-and-shape spaces), radar imaging~\cite{smith05intrinsic} (Grassman
manifold), diffusion tensor magnetic resonance imaging~\cite{lenglet05riemannian,lenglet06statistics}
%
\begin{comment}
\cite{arsigny07geometric,pennec99probabilities,fletcher07riemannian,dryden08noneuclidean,moakher05differential}
\end{comment}
{}(positive definite tensors), and Lie groups in general~\cite{grenander65algebraic}.
This section recalls the definition of Riemannian mean on manifolds;
the reader is assumed to be familiar with basic differential geometry
(e.g.,~\cite{docarmo94}).%
\begin{comment}
Consequently, from the mid 70s {[}karcher{]} until recent years {[}smith05{]},
there has been an effort in how to translate statistics concepts on
non-flat spaces. The notation and terminology has not stabilized so
far, therefore here we recall the main concepts. 
\end{comment}
{}

Let~$X$ be a random variable taking values in~$\mathbb{R}^{n}$
with joint cumulative distribution function~$\mu$. The expected
value of~$X$ (or Euclidean mean, or simply \emph{mean}) is defined,
in the most general terms, as the Lebesgue-Stieltjes integral\begin{equation}
\ex\{X\}\triangleq\int_{\mathbb{R}^{n}}\boldsymbol{x}\, d\mu(\boldsymbol{x}).\label{eq:mean1}\end{equation}
This definition is not directly generalizable to manifolds because
it assumes that the set has a vector space structure. However, the
mean satisfies a variational property, being the point that minimizes
the quadratic risk\begin{equation}
\ex\{X\}=\arg\min_{y}\ex\left\{ \left\Vert X-y\right\Vert _{2}^{2}\right\} .\label{eq:mean2}\end{equation}
Definitions~\eqref{eq:mean1} and \eqref{eq:mean2} are easily seen
to be equivalent in the case of vector spaces. The second definition
has the benefit that it can be generalized to any metric space.

In particular, it can be generalized to Riemannian manifolds. Recall
that a Riemannian manifold~$(\manifold,\metric)$ is a differentiable
manifold~$\manifold$ equipped with a smooth metric~$\metric$ on
the tangent space. The {}``length''~$\ell(\gamma)$ of a differentiable
curve $\gamma:[0,1]\rightarrow\manifold$ is defined by\[
\ell(\gamma)=\int_{0}^{1}\!\sqrt{\metric(\gamma'(t),\gamma'(t))}\, dt.\]
Given the notion of length, one defines the distance between two points
$m_{1},m_{2}\in\manifold$ as\[
\rdist(x,y)=\inf\{\ell(\gamma)\,|\,\gamma\,\mbox{is a differentiable curve joining }x\mbox{ and }y\}.\]
Consider a Riemannian manifold $(\manifold,\metric$) with corresponding
distance~$\rdist$. Generalizing~\eqref{eq:mean2} for a random
variable~$X$ taking values in~$\manifold$, define the Riemannian
mean (also called: Riemannian barycenter, Riemannian center of mass,
Frechét mean or Karcher mean) as the point that minimizes the average
quadratic distance: \begin{equation}
\exi\{X\}\triangleq\arg\inf_{y\in\manifold}\ex\left\{ \rdist^{2}(X,y)\right\} .\label{eq:definition-intrinsic-mean}\end{equation}
The Riemannian mean is unique for a simply connected manifold of non
positive sectional curvature~\cite{Karcher77} --- as counterexamples,
the reader may consider the distribution consisting of a pair of antipodal
points on the unit circle~$\mathbb{S}^{1}$ (a non-simply connected,
zero curvature manifold) and on the unit sphere~$\mathbb{S}^{2}$
(a simply connected, positive curvature manifold). See~\cite{corcuera99riemannianbarycentres}
for an alternative characterization of the Riemannian mean using the
inverse of the exponential map, and, more in general, see~\cite{smith05intrinsic}
for a short introduction to modern intrinsic estimation on manifolds.

\section{\label{sec:Different-metrics-for}Different metrics for the manifold
of Gaussian distributions}

The random availability of the observations in the LGB filtering setup
induces a stationary distribution on~$\Gauss(n)$, the manifold of
Gaussian distributions on~$\mathbb{R}^{n}$, for the estimate $\left(\hat{\vx},\mP\right)$.
In particular, the estimation error $\ve=\hat{\vx}-\vx$ has a Gaussian
distribution with zero mean, therefore we focus on the submanifold~$\GaussO(n)$
of Gaussian distributions with mean~$\boldsymbol{0}$. The Riemannian
mean depends on the choice of a metric, and this section considers
several such options. In the following, for compactness of notation,
sometimes we confound $\GaussO(n)$ with $\SymP(n)$, for example
by writing the distance between Gaussian distributions $d(\gaussdist(\vzero,\mP_{1}),\gaussdist(\vzero,\mP_{2}))$
directly as $d(\mP_{1},\mP_{2})$.

\subsection{Flat metric for covariances }

The traditional way to represent a Gaussian distribution is by using
its mean and covariance, by identifying~$\Gauss(n)$ with $\mathbb{R}^{n}\times\SymP(n)$.
This is what most consider to be the {}``natural'' representation.
Considering~$\SymP(n)$ as a convex cone of $\mathbb{R}^{n\times n}$
seems also to fit very well with the operation of expected value,
because the expectation is nothing other than a glorified convex linear
combination; the convexity is also useful in optimization, for example
in semidefinite programming~\cite{vanderberghe96semidefinite}. 

From this point of view, $\SymP(n)$ inherits the Euclidean metric
of~$\mathbb{R}^{n\times n}$. This is the metric implicitly used
by Sinopoli \emph{et al.~\cite{sinopoli04kalman}.} The distance
between two Gaussian distributions reduces to the Frobenius distance
between the two covariance matrices:\begin{equation}
d(\gaussdist(0,\mP_{1}),\gaussdist(0,\mP_{2}))=||\mP_{1}-\mP_{2}||_{F}.\label{eq:distfro}\end{equation}
For the Riemannian mean, one obtains that the covariance of the mean
distribution is the expected value of the covariances: $\exi\{\gaussdist(\vzero,\mP)\}=\gaussdist(\vzero,\ex\{\mP\})$.
Note that this mean is affine-invariant, i.e. invariant to a change
of coordinate $\mP\mapsto\mA\mP\mA^{*}$ for $\mA\in\mathrm{GL}(n)$,
but the distance~\eqref{eq:distfro}~is~not.

\subsection{Flat metric for information matrices}

The Information filter~\cite{maybeck79} utilizes a different parametrization
$(\veta,\mY)$ to represent a Gaussian distribution ($\veta=\mP^{-1}\hat{\vx}$
and $\mY=\mP^{-1}$); this gives a \emph{different} embedding of~$\Gauss(n)$
in $\mathbb{R}^{n}\times\SymP(n)$. One could make the argument that
information matrices are a more natural parametrization for Gaussian
distributions: in the canonical representation of Gaussian distributions
as an exponential family~\cite{CasellaBerger}, the information matrix
is the natural parameter; in fact, one writes the probability density
function using $\mP^{-1}$. With this choice, the distance between
distributions is given by\[
d(\gaussdist(0,\mP_{1}),\gaussdist(0,\mP_{2}))=||\mP_{1}^{-1}-\mP_{2}^{-1}||_{F}.\]
The Riemannian mean is $\exi\{\gaussdist(\vzero,\mP)\}=\gaussdist(\vzero,\ex\{\mP^{-1}\}^{-1})$.
It is easy to see that, in the Linear/Gaussian/Bernoulli case, $\exi\{\gaussdist(\vzero,\mP)\}$
exists for all values of $\probArr\in(0,1]$, because $\mP^{-1}$
is bounded by $\mPinf^{-1}$. As in the previous case, the mean is
affine-invariant, but the distance is not.

\subsection{Flat metric for square root of covariances}

The matrix equivalent of the scalar standard deviation is the square
root of the covariance matrix; the eigenvalues of $\sqrt{\mP}$ are
the standard deviations. The distance between distributions can be
defined as\[
d(\gaussdist(0,\mP_{1}),\gaussdist(0,\mP_{1}))=||\sqrt{\mP_{1}}-\sqrt{\mP}_{2}||_{F},\]
and consequently the Riemannian mean is $\exi\{\gaussdist(\vzero,\mP)\}=\gaussdist(\vzero,\ex\{\sqrt{\mP}\}^{2})$.
By Jensen's inequality and the fact that $\mP\mapsto\sqrt{\mP}$ is
operator-concave, it follows that $\ex\{\sqrt{\mP}\}^{2}\leq\ex\{\mP\}$.
Thus the Riemannian mean for this distance exists in all cases when
$\ex\{\mP\}$ exists; moreover, as shown by Proposition~\ref{pro:1},
in some cases, the critical probability for boundedness of $\ex\{\sqrt{\mP}\}$
is strictly less than the critical probability for $\ex\{\mP\}$.

\subsection{Fisher Information Metric}

The problem we are analysing is special in two regards:\textbf{ }1)~we
are concerned with doing statistics on a certain manifold~$\Gauss(n)$;
and 2)~the elements of the manifold~$\Gauss(n)$ represents probability
distribution themselves. The branch of statistics that studies the
properties of the families of probability distributions considered
as a manifold is called \emph{information geometry}%
\begin{comment}
%
\footnote{Not to be confounded with \emph{stochastic geometry}, which studies
the statistics of geometric shapes.%
}
\end{comment}
{} and is a relatively recent development with respect to classical
mathematical statistics~\cite{amari85,murray93differential}%
\begin{comment}
Amari 00
\end{comment}
{}. 

From this point of view, the manifold has a natural metric given by
the generalization of the Fisher Information Matrix (FIM) as a Riemannian\emph{
}metric. We recall the definition of the FIM in the Gaussian case~\cite{CasellaBerger}.
Assume that the available observations $\vz\in\mathbb{R}^{q}$ have
a Gaussian distribution whose mean and covariance are parametrized
by an unknown parameter~$\vtheta\in\mathbb{R}^{n}$: $\vz\sim\gaussdist(\vzmean(\vtheta),\vzcov(\vtheta))$.
The FIM for~$\vtheta$ is the $n\times n$ semidefinite positive
matrix $\fim[\vtheta]$ defined as\begin{equation}
\fim[\vtheta]_{a,b}=\frac{\partial\vzmean}{\partial\theta_{a}}^{*}\vzcov(\theta)^{-1}\frac{\partial\vzmean}{\partial\theta_{b}}+\frac{1}{2}\mbox{Tr}\left\{ \vzcov(\theta)^{-1}\frac{\partial\vzcov(\theta)}{\partial\theta_{a}}\vzcov(\theta)^{-1}\frac{\partial\vzcov(\theta)}{\partial\theta_{b}}\right\} .\label{eq:fim}\end{equation}
The FIM gives the information contained in the samples about the value
of~$\vtheta$; for example, using the FIM one defines the Cramér-Rao
Bound for unbiased estimators as $\mathrm{cov}(\hat{\vtheta})\geq\fim[\vtheta]^{-1}$.
The FIM can be generalized to be a Riemannian metric for the manifold~$\Gauss(n)$.
If we restrict to the submanifold $\GaussO(n)$, given two elements
$\mX,\mY\in\Sym(n)$ in the tangent space at $\gaussdist(0,\mP)$,
the Fisher Information Metric is \begin{equation}
g\left(\mX,\mY\right)=\frac{1}{2}\mbox{Tr}\left\{ \mP^{-1}\mX\mP^{-1}\mY\right\} .\label{eq:metric_0mean}\end{equation}
Compare~\eqref{eq:metric_0mean} with the second term in~\eqref{eq:fim}.
%
\begin{comment}
It is shown in Maass {[}19, Thm. p. 27{]} (see also Tetras {[}27{]}),
that 8 is the usual Riemannian distance on o when this set is considered
as the Riemannian symmet
\end{comment}
{}The distance induced by this metric is%
\begin{comment}
%
\footnote{Most of the statistics literature uses $\mP_{2}^{-1/2}\mP_{1}\mP_{2}^{-1/2}$
instead of $\mP_{1}\mP_{2}^{-1}$; note that these two matrices have
the same eigenvalues, which are the solutions to the generalized eigenvalue
problem for the couple $(\mP_{1},\mP_{2})$. %
}
\end{comment}
{}%
\begin{comment}
~\cite{talih07geodesic}
\end{comment}
{}~(e.g., \cite{lenglet05riemannian}):\begin{equation}
d(\gaussdist(\vzero,\mP_{1}),\gaussdist(\vzero,\mP_{2}))=\left[\sum_{i=1}^{n}\log^{2}\left(\lambda_{i}(\mP_{1}\mP_{2}^{-1})\right)\right]^{1/2}\label{eq:rdist}\end{equation}
This distance is {}``natural'' in the sense that it is linked to
the probability of distinguishing the two distributions $\gaussdist(\vzero,\mP_{1})$
and $\gaussdist(\vzero,\mP_{2})$ by observing their samples, in a
sense which is made precise in~\cite{amari85}. Unfortunately, a
closed form expression for writing $\exi\{\gaussdist(\vzero,\mP)\}$
is not known. 

The use of this natural distance on $\GaussO(n)\simeq\SymP(n)$ allows
to show that some naive results obtained using the flat metric on
covariance are incorrect~\cite{smith05intrinsic}. For example, in
basic mathematical statistics courses, one teaches that, given a set
of samples $\{\vx_{i}\}_{i=1}^{n}$ from a distribution with covariance~$\mP$,
the bias-corrected sample covariance matrix $\hat{\mP}_{sc}=\frac{1}{n-1}\sum_{i=1}^{n}(\vx-\vx_{i})(\vx-\vx_{i})^{*}$
is an \emph{unbiased} and \emph{efficient} estimator of~$\mP$, i.e.
$\ex\{\hat{\mP}_{sc}\}=\mP$ and $\hat{\mP}_{sc}$ reaches the Cramér-Rao
bound. However, it is also well known that $\hat{\mP}_{sc}$ performs
poorly at low sample support. Smith's~\cite{smith05intrinsic} explanation
to this conundrum is that $\hat{\mP}_{sc}$ is not unbiased according
to the natural metric: $\exi\{\hat{\mP}_{sc}\}\neq\mP$. 

Ignoring the Fisher Information Metric interpretation, the distance
defined by~\eqref{eq:rdist} is also natural for $\SymP(n)$ when
it is considered either as a symmetric space, or as the quotient space~$\mathrm{GL}(n)/\mathrm{O}(n)$~\cite{smith05intrinsic}.
The distance has several other useful properties in the context of
LGB filtering. $(\SymP(n),\rdist)$ is a complete metric space~\cite{bougerol93kalman}
and a geodesically complete manifold with nonpositive curvature~\cite{lenglet06statistics}.
The distance~$\rdist$ induces the usual topology~\cite{bougerol93kalman}
on~$\SymP(n)$. The distance is invariant to affine transformations,
and also to inversion $\mP\mapsto\mP^{-1}$; this last property is
useful because one can use either covariance matrices or information
matrices: $\exi\{\mP\}=\exi\{\mP^{-1}\}^{-1}$, which is not true
if one uses the expected value. Using this distance it is also easy
to show contraction properties for the Riccati iterations $g$, $h$
that guarantee the existence of the stationary distribution~\cite{bougerol93kalman}. 

It is possible to show that there is no {}``critical probability''
if one uses this metric. To this end, one should first prove that
the system has a stationary distribution for all values of $\probArr>0$.
This is done in the next section. Then, in Section~\ref{sub:Existence-of-the},
it is proved that the Riemannian mean of this distribution exists.

\subsection{Existence of the stationary distribution}

In this section, we prove the existence of the stationary distribution,
for every value of $\probArr>0$. This can be done by using some results
from Bougerol~\cite{bougerol93kalman} regarding the contraction
properties of the maps~$h$ and~$g$, and some results from Barnsley
\emph{et al.~}\cite{barnsley88new} about the convergence of Iterated
Function Systems%
\footnote{Because this paper is not available electronically yet, the results
are stated here extensively.%
}. Once these results are recalled, the conclusion will be immediate.

\medskip

We need some preliminaries from~\cite{barnsley88new}. Let $(X,d)$
be a complete metric space. Let $f_{i}$, $i=1,\dots,n$ be Lypschitz
functions from~$X$ to~$X$, that is, there exists $s_{i}>0$ such
that $d(f_{i}(x),f_{i}(y))\leq s_{i}d(x,y)$ for all $x,y$ in $X$.
We say that $f_{i}$ is {}``nonexpansive'' if $s_{i}\leq1$, and
we say that it is a {}``strict contraction mapping'' if $s_{i}<1$.
Assign a set of probabilities~$p_{i}$, to these functions, such
that $p_{i}>0$ and $\sum_{i=1}^{n}p_{i}=1$. Consider now the Iterated
Function System $\{(f_{i},p_{i})\}$ and the corresponding Markov
chain, which we denote $\{Z_{i}\},i\geq0$. We say that a measure
$\mu$ is {}``attractive'' if, for every initial distribution of~$Z_{0}$,
the process~$Z_{i}$ converges in distribution to~$\mu$, that is,
$\lim_{i\rightarrow\infty}\mathbb{E}\{f(Z_{i})\}=\int f\, d\mu$ for
every bounded continuous function~$f$ on~$X$. It is an intuitive
result that, if all the~$f_{i}$ are strict contractions, then the
process {}``tends to forget'' the initial conditions, and a stationary
distribution exists. What is not trivial is that IFSs converge in
distribution with much weaker hypotheses, as shown by the following
result.
\begin{thm}
\label{thm:(Barnsley-et-al)}(Barnsley \emph{et al}.~\cite{barnsley88new})
Suppose that the $f_{i}$ satisfies an \emph{average contractivity
condition} as follows: for all $x,y\in X$,\begin{equation}
\sum_{i=1}^{n}p_{i}\log\frac{d(f_{i}(x),f_{i}(y))}{d(x,y)}<0.\label{eq:average}\end{equation}
Then there exists a unique, attractive invariant probability measure
for the IFS.\end{thm}
\begin{remrk}
Note that it is not assumed that the single $f_{i}$ are strict contractions
($s_{i}<1$) or even contractions $(s_{i}\leq1)$. This theorem can
also be generalized to the case in which the transition probabilities
depend on the state ($p_{i}=p_{i}(x)$), although some additional
hypotheses are required~\cite{barnsley88place}.
\end{remrk}
\medskip

We now recall the following from~\cite{bougerol93kalman}:
\begin{lemma}
(Bougerol~\cite{bougerol93kalman}) \label{lem:boug-lemma}In the
metric $d$ defined by~\eqref{eq:rdist},
\begin{enumerate}
\item The maps $h$ and $g$ are nonexpansive mappings: $d(h(\mP_{1}),h(\mP_{2}))\leq d(\mP_{1},\mP_{2}),$
and equivalently for $g$.
\item If $\mA$ is non-singular, $(\mA,\mB)$ is controllable, $(\mA,\mC)$
observable, the composition $g^{n}=g\circ\cdots\circ g$ of~$n$
copies of~$g$ is a strict contraction mapping; that is, there exists
$\rho=\rho(\mA,\mB,\mC)<1$ such that $d(g(\mP_{1}),g(\mP_{2}))\leq\rho d(\mP_{1},\mP_{2})$.
\end{enumerate}
\end{lemma}
\medskip

From these results, the following result is easily proved.
\begin{prop}
\label{pro:stationary-exists}If $\mA$ is non-singular, $(\mA,\mB)$
controllable, $(\mA,\mC)$ observable, then the stationary distribution
for~$\mP$ exists for all $\probArr>0$.\end{prop}
\begin{proof}
Consider the behavior of the system at intervals of~$n$ steps. This
corresponds to considering the {}``power'' IFS $\ifs^{n}=\{(g^{n},\probArr^{n}),(g\circ h^{n-1},\probArr(1-\probArr)^{n-1}),\dots,(h^{n},(1-\probArr)^{n})\}$,
which is created by all~$2^{n}$ possible combinations of length~$n$
of the functions~$g,$~$h$, with corresponding probabilities. By
Lemma~\ref{lem:boug-lemma}, $g^{n}$ is a strict contraction, and
all the other combinations are nonexpansive mappings. Therefore, assuming
$\probArr>0$, the system satisfies the average contractivity condition~\eqref{eq:average},
and by Theorem~\ref{thm:(Barnsley-et-al)} the stationary distribution
exists.
\end{proof}

\subsection{Existence of the Riemannian mean for Fisher Information Metric\label{sub:Existence-of-the}}

After having ascertained that the stationary distribution exists (Proposition~\ref{pro:stationary-exists}),
we now prove existence of the Riemannian mean.
\begin{prop}
\label{pro:rdist-bound-1}The Riemannian mean of the LGB distribution
for the distance~\eqref{eq:rdist} exists for all~$\probArr\in(0,1]$.\end{prop}
\begin{proof}
The stationary distribution of the IFS $\ifs=\{(\fg,\probArr),(\fh,1-\probArr)\}$
is equivalent to that of the power IFS $\ifs^{n}=\{(g^{n},\probArr^{n}),(g\circ h^{n-1},\probArr(1-\probArr)^{n-1}),\dots,(h^{n},(1-\probArr)^{n})\}$,
obtained by considering compositions of length~$n$ of the functions
$(g,h)$ with corresponding probabilities. We now build the IFS $\ifs_{\pess}^{n}$,
a {}``pessimist'' approximation to $\ifs^{n}$. By recalling that~$\fg$
and~$\fh$ are order-preserving, and $\fg(\mM)\leq\fh(\mM)$ for
all~$\mM$~\cite{sinopoli04kalman}, one can bound all mixed terms
in $\fg$,~$\fh$ in $\ifs^{n}$ by~$h^{n}$. Furthermore, one can
also find an upper bound for~$\fg^{n}$: because the system is observable,
the uncertainty is bounded over all the state space after~$n$ consecutive
observation are received. Therefore, $\mWpess\triangleq\sup_{\mP\geq\mPinf}\fg^{n}(\mP)$
exists and is bounded: $\mPinf\leq\mWpess<\infty$. Thus the pessimist
approximation to~$\ifs^{n}$ is $\ifs_{\pess}^{n}=\{(\mWpess,\probArr^{n}),(h^{n},1-\probArr^{n})\}$.
Seeing the stationary variables $\mP$ (for $\ifs^{n}$) and $\mPpess$
(for $\ifs_{\pess}^{n}$) as functions of the past infinite sequence
of arrivals $\astring=\{\arr(0),\arr(-1),\arr(-2),\dots\}\in\stringset$,
one has that \begin{equation}
\mPinf\leq\mP(\astring)\leq\mPpess(\astring).\label{eq:order}\end{equation}

To prove that $\exi\{\mP\}$ is bounded for all~$\probArr$, it is
sufficient to show that the minimization problem~\eqref{eq:definition-intrinsic-mean}
is feasible for all~$\probArr$. To prove this, it is sufficient
to show that $\ex\{d^{2}(\mX,\mP)\}$ is bounded for some matrix~$\mX$;
it is convenient to choose $\mX=\mPinf$. By~\eqref{eq:order} and
Lemma~\ref{lem:dist} below, we obtain that $\rdist(\mPinf,\mP(\astring))\leq\rdist(\mPinf,\mPpess(\astring))$
and thus $\ex\{d^{2}(\mPinf,\mP)\}\leq\ex\{d^{2}(\mPinf,\mPpess)\}$.
It follows that $\exi\{\mP\}$ is bounded if $\exi\{\mPpess\}$ is. 

We now investigate boundedness of $\ex\{d^{2}(\mPinf,\mPpess)\}$.
Choose an $\mM$ such that $\mM\mPinf\mM^{*}=\mI$ and do a change
of coordinates $\mP\mapsto\mM\mP\mM^{*}$. One finds that $\ex\{d^{2}(\mPinf,\mPpess)\}=\ex\{d^{2}(\mI,\mPpess')\}$,
where $\mPpess'\triangleq\mM\mPpess\mM^{*}\geq\mI$. Note that because
$\lambda_{i}(\mPpess')\geq1$, we can find the bound $d^{2}(\mI,\mPpess')=\sum_{i=1}^{n}\log^{2}\left(\lambda_{i}(\mPpess'\right)\leq n\,\log^{2}\left(\left\Vert \mPpess'\right\Vert \right)$.
An expression for~$\mPpess(k)$ can be written explicitly as in the
proof of Proposition~\ref{pro:1} as a function of $\tau(k)$, the
number of steps passed without receiving an observation (recall that
one time step in the IFS $\ifs^{n}$ corresponds to~$n$ steps of~$\ifs$):\[
\mPpess(k)=\mA^{n\tau(k)}\mWpess\left(\mA^{*}\right)^{n\tau(k)}+\sum_{i=0}^{n\tau(k)-1}\mA^{i}\mQ\left(\mA^{*}\right)^{i}.\]
From this one finds the bound $\left\Vert \mPpess(k)\right\Vert \leq c_{1}\left\Vert \mA\right\Vert ^{2n\tau(k)}$
for some $c_{1}>0$. Thus $n\log^{2}\left(\left\Vert \mPpess'\right\Vert \right)\leq\tau(k)^{2}c_{2}+\tau(k)c_{3}+c_{4}$
for some $c_{2},c_{3},c_{4}>0$. As in the proof of Proposition~\ref{pro:1},
the expectation can be computed with respect to $\tau(k)$, and one
obtains \[
\ex\{d^{2}(\mPinf,\mP)\}\leq\probArr^{n}\sum_{i=0}^{\infty}(1-\probArr^{n})^{i}\left[i^{2}c_{2}+ic_{3}+c_{4}\right].\]
\vspace{-1.5mm}Series of the kind $\sum_{i=0}^{\infty}i^{k}x^{i}$
with $k\geq0$ are convergent if $|x|<1$, hence $\ex\{d^{2}(\mPinf,\mP)\}$
is always bounded if $\probArr\in(0,1]$. Therefore, the Riemannian
mean is always bounded. \end{proof}
\begin{lemma}
For the distance defined in~\eqref{eq:rdist}, \label{lem:dist}$\mP_{1}\leq\mP_{2}\leq\mP_{3}\,\,\Rightarrow\,\,\rdist(\mP_{1},\mP_{2})\leq\rdist(\mP_{1},\mP_{3})$.\end{lemma}
\begin{proof}
(sketch) First, reduce to the case $\mP_{1}=\mI$ by letting $\mP'_{i}=\mM\mP_{i}\mM^{*}$,
with $\mM$ chosen such that $\mM\mP_{1}\mM^{*}=\mI$. Then verify
$\rdist(\mI,\mP'_{2})\leq\rdist(\mI,\mP'_{3})$ by direct computation
using~\eqref{eq:rdist}. 
\end{proof}
After one has proved the existence of $\exi\{\mP\}$, uniqueness follows
from the fact that the manifold has nonpositive curvature~\cite{Karcher77}.

%
\begin{comment}

\subsubsection{Information-theoretic distances}

Finally, it is worth mentioning other information-theoretical distances
which are not. The KLD~\cite{kld51} is defined by:\[
KLD\]

Bregman divergence is the average

This is not a distance per se, because it is not symmetric, but it
can be symmetrized in various ways. $d(\mP_{1},\mP_{2})=$Mu

Source coding vs 
\end{comment}
{}

%\nocite{*}

\vspace{-1mm}

\section{Conclusions}
\begin{quotation}
Algebra is the offer made by the devil to the mathematician. The devil
says: \textquoteleft{}I will give you this powerful machine, it will
answer any question you like. All you need to do is giving me your
soul: \emph{give up geometry} and you will have this marvellous machine.\textquoteright{}~\cite{mathematical_evolutions}

\hfill\textemdash{} Sir Michael Atiyah (1929--)
\end{quotation}
The righteous engineer must refuse the devil's offer. Reframing problems
in a geometric framework usually allows to spot the hidden assumptions,
and therefore to check whether the results have a physical meaning,
or they are just figments of the mathematical formalization.

The hidden assumption in the work of Sinopoli \emph{et al. }is that
positive definite matrices are treated as a convex cone of $\mathbb{R}^{n\times n}$.
This is perhaps the most intuitive interpretation, and has good consequences
in certain contexts, such as semidefinite programming~\cite{vanderberghe96semidefinite}.
However, this could lead to incorrect\emph{ }conclusions when doing
rigorous intrinsic estimation\emph{. }This is well shown by Smith's
example~\cite{smith05intrinsic}, that even though the sample covariance
matrix is unbiased in the naive sense ($\ex\{\hat{\mP}_{sc}\}=\mP$),
it is biased in the intrinsic sense ($\exi\{\hat{\mP}_{sc}\}\neq\mP$),
thereby contradicting what is taught in elementary statistics courses. 

In the case of the Linear/Gaussian/Bernoulli filtering problem, if
one uses the average standard deviations, instead of the average covariances,
one obtains a different critical probability~(Proposition~\ref{pro:1}).
It is pointless to discuss which critical probability is more critical
than the other, but surely considering the average error is more natural
than considering the average error \emph{squared}. The point is that
the boundedness of the expected value cannot be considered as a criterion
for {}``stability''; there are plenty of well-behaved probability
distributions which have infinite moments%
\footnote{ Consider as an example the Pareto distribution, defined as $\pr(\{X>x\})=x^{-k}$
for $x\geq1,k>0$: for this distribution, $\ex\{X\}$ is bounded only
if $k>1$, and $\ex\{X^{2}\}$ only if $k>2$; yet it has a regular
power law. All the statistics that remain finite change smoothly when~$k$
goes through the {}``critical'' values~$1$ and~$2$.%
}. 

Sinopoli's critical probability is critical only in the sense that
it is the threshold under which $\ex\{\mP\}$ ceases to be meaningful
as a performance measure. Under the threshold, the qualitative behavior
of the system does not change, as one can see from considering the
Riemannian mean derived from the intrinsic Fisher Information Metric
(Proposition~\ref{pro:rdist-bound-1}) --- that is possibly a less
intuitive, but more natural way to represent the concept of {}``average
uncertainty''. Thus the intense research effort in trying to characterize~$\probArrCrit$
seems misplaced; of more interest is studying the entire LGB distribution~\cite{censi09fractals,epstein08probabilistic}.

%
\begin{comment}
Another very interesting point of view is the one of Tarantola, on
the invariance of Choosing metric == choosing basic measure (tarantola)
\end{comment}
{}

%
\begin{comment}
median instea d of mean
\end{comment}
{}

%
\begin{comment}
Note that, throughout the history of engineering, the use of quadratic
forms were easy because they lead to computations. This was at the
basis of the definition of Gaussian (from gauss to Kalman {[}?{]})
and the reason of LQR, LQG problems in control theory. 
\end{comment}
{}

\bibliographystyle{ieeetr}
\bibliography{bib/notflat}

\stop

\cleardoublepage
\begin{prop}
\label{pro:rdist-bound}Let $\fim>0$. With the distance defined by~\eqref{eq:rdist},
the Riemannian mean of the LGB distribution is bounded for all~$\probArr\in(0,1)$.\end{prop}
\begin{proof}
To prove that $\exi\{\mP\}$ is bounded for all $\probArr$, it is
sufficient to show that the minimization problem $\arg\inf_{\mM}\ex\{d^{2}(\mM,\mP)\}$
is feasible for all $\probArr$, with the distance defined by~\eqref{eq:rdist}.
To prove this, it is sufficient to show that $\ex\{d^{2}(\mM,\mP)\}$
is bounded for some matrix $\mM$; it is convenient to choose $\mM=\mI$.
Then existence and uniqueness of the Riemannian means follows from
the fact that the manifold has nonpositive curvature. We start by
the obvious bound $\ex\{d^{2}(\mI,\mP)\}\leq n\ex\left\{ \log^{2}(\lambda_{\max}(\mP))\right\} $.
An expression for $\lambda_{\max}(\mP)$ can be found by considering
the same {}``pessimist'' system in the proof of Proposition~\ref{pro:1},
this time written for the full matrix system:\[
\mPpess(k)=\begin{cases}
\mWpess & \tau(k)=0\\
\mA^{\tau(k)}\mWpess\left(\mA^{*}\right)^{\tau(k)}+\sum_{i=0}^{\tau(k)-1}\mA^{i}\mQ\left(\mA^{*}\right)^{i} & \tau(k)\geq1\end{cases}\]
with $\mWpess\triangleq\fim^{-1}$. Straightforward computations lead
to:\begin{eqnarray*}
\lambda_{\max}\mPpess(k) & \leq & \lambda_{\max}(\mWpess)\left(\sigma_{\max}\mA\right)^{2\tau(k)}+\lambda_{\max}(\mQ)\sum_{i=0}^{\tau(k)-1}\left(\sigma_{\max}\mA\right)^{2i}\\
 & \leq & \left(\sigma_{\max}\mA\right)^{2\tau(k)}\left[\lambda_{\max}(\mWpess)+\frac{\lambda_{\max}(\mQ)}{(\sigma_{\max}\mA)^{2}-1}\right]\end{eqnarray*}
This implies that $\log\lambda_{\max}(\mPpess(k))\leq\tau(k)c_{1}+c_{2}$
for some $c_{1},c_{2}>0$. It follows that, for some $a,b,c>0$, $\log^{2}(\lambda_{\max}(\mPpess))=\tau(k)^{2}a+\tau(k)b+c$.
The expectation can be computed as\[
\ex\left\{ \log^{2}(\lambda_{\max}(\mP))\right\} \leq\probArr\sum_{i=0}^{\infty}(1-\probArr)^{i}\left[i^{2}a+ib+c\right]\]
Series of the kind $\sum_{i=0}^{\infty}i^{k}x^{i}$ with $k\geq0$
are convergent if $|x|<1$, hence $\ex\{d^{2}(\mI,\mP)\}$ is always
bounded. Therefore, the Riemannian mean is always bounded.  
\end{proof}

\end{document}